\documentclass[letterpaper, 10 pt, conference]{ieeeconf}  

\IEEEoverridecommandlockouts                              
\let\labelindent\relax
\usepackage[normalem]{ulem}
\usepackage{graphicx}
\usepackage{amsmath,amssymb,amsfonts,mathabx}
\usepackage{float}
\usepackage{cite}
\usepackage[colorlinks = true, linkcolor = blue, citecolor =
blue]{hyperref}
\usepackage{enumitem}
\usepackage{xcolor}

\graphicspath{{matlab/},{./}}

\newcommand{\real}{\ensuremath{\mathbb{R}}}
\newcommand{\realn}{\ensuremath{\mathbb{R}^n}}

\newcommand{\realnonneg}{\ensuremath{\mathbb{R}_{\ge 0}}}
\newcommand{\identityn}{\mathbf{I}_n}

\newcommand{\nat}{\ensuremath{\mathbb{N}}}

\newcommand{\Mc}[1]{\mathcal{#1}}
\newcommand{\Vi}{v_i}
\newcommand{\boldv}{\mathbf{v}}
\newcommand{\V}{v}

\newcommand{\agt}{\mathcal{V}} 

\newcommand{\agtnotj}{\mathcal{V}_{-j}}

\newcommand{\wik}{w_{ik}}
\newcommand{\wki}{w_{ki}}
\newcommand{\wkl}{w_{kl}}

\newcommand{\wij}{w_{ij}}

\newcommand{\wji}{w_{ji}}
\newcommand{\wijk}{w_{ij}^{(k)}}
\newcommand{\wjk}{w_{jk}}
\newcommand{\yij}{y_{ij}}

\newcommand{\boldxi}{\mathbf{x}_i}
\newcommand{\boldzi}{\mathbf{z}_i}
\newcommand{\bolddi}{\mathbf{d}_i}
\newcommand{\boldqi}{\mathbf{q}_i}
\newcommand{\boldwi}{\mathbf{w}_i}
\newcommand{\boldwik}{\mathbf{w}_{i_k}}
\newcommand{\boldyi}{\mathbf{y}_i}
\newcommand{\boldwistar}{\mathbf{w}_i^*}

\newcommand{\boldw}{\mathbf{w}}
\newcommand{\boldx}{\mathbf{x}}
\newcommand{\boldy}{\mathbf{y}}
\newcommand{\boldwstar}{\mathbf{w}^*}

\newcommand{\boldwnoti}{\mathbf{w}_{-i}}
\newcommand{\boldwnotik}{\mathbf{w}_{-i_k}}
\newcommand{\boldwnotistar}{\mathbf{w}_{-i}^*}

\newcommand{\grph}{\mathsf{G}}
\newcommand{\negrph}{\mathsf{G}^*}
\newcommand{\grphunweighted}{\mathsf{G}^\dagger}
\newcommand{\edg}{\mathcal{E}}
\newcommand{\edgunweighted}{\mathcal{E}^\dagger}
\newcommand{\adjmat}{\mathbf{A}}
\newcommand{\adjmatunweighted}{\mathbf{A}^\dagger}
\newcommand{\Ni}{\mathcal{N}_{i}}

\newcommand{\Ninoti}{\mathcal{N}_{i}\setminus\{i\}}

\newcommand{\agtsumjni}{\sum_{j \in \Ni}}

\newcommand{\agtsumk}{\sum_{k \in \agt}}
\newcommand{\agtsumj}{\sum_{j \in \agt}}
\newcommand{\agtsumknoti}{\sum_{k \in \agt_{-i}}}
\newcommand{\agtsumlnoti}{\sum_{l \in \agt_{-i}}}

\newcommand{\constraintseti}{\mathcal{K}_i}
\newcommand{\constraintsetik}{\mathcal{K}_{i_k}}
\newcommand{\constraintset}{\mathcal{K}}
\newcommand{\constraintsetnoti}{\mathcal{K}_{-i}}
\newcommand{\katzscorei}{c_i}
\newcommand{\katzscorel}{c_l}
\newcommand{\katzscoreik}{c_{i_k}}
\newcommand{\katzscorej}{c_j}
\newcommand{\katzscorek}{c_k}
\newcommand{\katzscorestacked}{\mathbf{c}}
\newcommand{\agtsumjniarg}[1]{\sum_{j \in \Ni({#1})}}

\newcommand{\inftysumk}{\sum_{k=1}^{\infty}}

\newcommand{\inftysumm}{\sum_{m=1}^{\infty}}
\newcommand{\ntwgame}{\mathcal{G}}
\newcommand{\neset}{\mathcal{NE}}
\DeclareMathOperator*{\argmax}{argmax}

\newcommand{\boldone}{\mathbf{1}}
\newcommand{\boldzero}{\mathbf{0}}
\newcommand{\boldei}{\mathbf{e}_i}
\newcommand{\boldej}{\mathbf{e}_j}

\newcommand{\mwalksumjtonoti}{S'_{j,m,i}}
\newcommand{\mwalksumjtoi}{S_{j,m,i}}
\newcommand{\pjiboldwnoti}{p_{ji}(\boldwnoti)}
\newcommand{\qjiboldwnoti}{q_{ji}(\boldwnoti)}
\newcommand{\bestrespseti}{\mathcal{BR}_i(\boldwnoti)}
\newcommand{\betterrespseti}{\mathcal{R}_i(\boldwi,\boldwnoti)}
\newcommand{\strictbetterrespseti}{\mathcal{R}_i^s(\boldwi,\boldwnoti)}
\newcommand{\diag}{\mathsf{diag}}
\newcommand{\supp}{\sigma}
\newcommand{\thmtitle}[1]{\mbox{}\textit{(\textbf{#1}.)}}

\newcommand{\remend}{\relax\ifmmode\else\unskip\hfill\fi\hbox{$\bullet$}}

\newcommand{\bulletsym}{\hbox{$\bullet$}}
\newcommand{\bulletend}{\relax\ifmmode\else\unskip\hfill\fi\bulletsym}
\newcommand{\squaresym}{\hbox{$\blacksquare$}}
\newcommand{\proofend}{\relax\ifmmode\else\unskip\hfill\fi\squaresym}
\newcommand{\trianglesym}{\hbox{$\blacktriangle$}}
\newcommand{\egend}{\relax\ifmmode\else\unskip\hfill\fi\trianglesym}

\renewenvironment{proof}{\textit{Proof:} }{\proofend}

\newtheorem{theorem}{\textbf{Theorem}}[section]
\newtheorem{corollary}[theorem]{\textbf{Corollary}}
\newtheorem{lemma}[theorem]{\textbf{Lemma}}

\newtheorem{remark}[theorem]{\textbf{Remark}}

\newcommand{\tth}{^{\text{th}}}

\renewcommand{\baselinestretch}{0.994}
\title{\LARGE \bf
	A Network Formation Game for Katz Centrality Maximization: A Resource Allocation Perspective
}

\author{Balaji R$^{1}$, Prashil Wankhede$^{1}$ and Pavankumar Tallapragada$^{1}$
	\thanks{ This work was partially supported by Science and Engineering Research Board under grant CRG/2023/008573.} \thanks{$^1$ Balaji R, Prashil Wankhede and Pavankumar Tallapragada are with the Indian Institute of Science, Bengaluru, India. \{\tt\small \{rbalaji, prashilw, pavant\}@iisc.ac.in\}}} 

\begin{document}

\maketitle
\thispagestyle{empty}
\pagestyle{empty}

\begin{abstract}
	In this paper, we study a network formation game in which agents seek to maximize their influence by allocating constrained resources to connections with other agents. We use Katz centrality to model agents' influence in the network. Allocations are restricted to neighbors in a given unweighted network, encoding topological constraints. The allocations by an agent correspond to the weights of its outgoing edges. Such allocation by all agents thereby induces a network. This models a strategic-form game in which agents’ utilities are given by their Katz centralities. We characterize the Nash equilibrium networks of this game and analyze their properties. We propose a sequential Best Response Dynamics (BRD) to model the network formation process. We show that it converges to the set of Nash equilibria under the assumptions of bounded budgets and that every agent does best response infinitely often. For complete underlying topologies, we show that Katz centralities are proportional to agents’ budgets at Nash equilibria. For general underlying topologies, in which each agent has a self-loop, we show that hierarchical networks form at Nash equilibria. Finally, simulations illustrate our findings.
\end{abstract}
\begin{keywords}
	Network formation games, Network Centrality maximization, Resource allocation, Best response dynamics.
\end{keywords}

\section{INTRODUCTION}

Many real-world systems are fundamentally networked, with outcomes determined not only by agents’ actions but also by the structure of interactions among them; importantly, this structure often emerges endogenously through strategic decisions. Network formation games provide a natural framework to study how self-interested agents choose connections and allocate limited resources to maximize utility, which may depend on access to information, influence, or resources. This perspective is particularly relevant in social and information networks, financial systems, and communication or collaboration networks, where links are formed under costs and constraints. In this work, we propose a network formation game for Katz centrality maximization. We view the game as one of strategic resource allocation, where agents allocate resources to their neighbors for activities such as investment, social influence, advertising, communication etc.

\subsubsection*{Literature review}
Network formation games have been widely studied over the years across a variety of application domains. For instance, the work~\cite{JacksonWolinsky1996} analyzes the stability of social and economic networks from a game theoretic perspective. A non-cooperative model of social network formation, in which agents form connections based on a trade-off between the rewards and costs of forming and severing links, was proposed in~\cite{BalaGoyal2000}. Research~\cite{cisneros2021network} proposes a network formation game to study the emergence of hierarchical networks in groups containing \emph{consensual} and \emph{non-consensual} agents. The paper~\cite{2005_Jackson_survey} surveys models of undirected
network formation and also studies their structure and efficiency. The works~\cite{BlochJackson2006,CalvoArmengol2009} study \emph{pairwise stability} of Nash equilibrium in network formation games.

Network centrality measures are used to quantify the importance or influence of individual nodes within a network. We refer the reader to \cite{2010_NL_BF_JH, 2023_FB_OMJ_PT} for definitions and introductions to commonly used centrality measures. The work in \cite{Salonen2016} considers a network formation game, wherein each player's utility function is a weighted sum of ``Cobb-Douglas'' functions and the weights are commonly agreed valuations of the players. For this game, the paper analyzes the relationship between the Nash equilibria and various centrality measures. In the present work, we consider the well-known \emph{Katz} centrality, originally proposed in \cite{Katz1953}, to model agents' utilities in the network formation game. Katz centrality finds widespread application in various domains, such as influence maximization in social networks \cite{2020_AS_BM}, consensus protocols \cite{2018_MP_OK_JH}, opinion dynamics \cite{2020_SD_WA_TJ_EA}, the characterization of Nash equilibria in Cournot games \cite{2019_KB_SE_RI}, and online social networks \cite{2015_AM_EA_YH}. Furthermore, a control-theoretic perspective of Katz centrality was proposed in \cite{2017_KS}.

In the theme of centrality-maximization in network formation games, ~\cite{2012_RT_CG} proposes a game-theoretic model in which each agent aims to maximize its relative Katz centrality while incurring costs for link formation. The works~\cite{Bei2011} and~\cite{LAOUTARIS20141266} study network formation models in which agents have budget constraints and seek to maximize their betweenness centrality, and seek to minimize the sum of preference-weighted distances to other nodes, respectively. The paper~\cite{cisneros2021network} studies a network formation game in which agents' utilities depend on their (out) degree centralities.

Dynamic models of network formation that employ BRD include~\cite{Konig2010, Castaldo2020} for Bonacich centrality maximization, ~\cite{Catalano2024} for a finite potential game for PageRank centrality maximization, in which asynchronous BRD converges to a Nash equilibrium in finite time. The work~\cite{Feldman2020} studies the effect of deviator rules on the efficiency of Nash equilibria reached under BRD in network formation games. Research~\cite{bazenkov2015} uses the so-called \emph{double best-response} dynamics to model network formation in wireless networks, showing that it generates more efficient networks compared to algorithms based on standard best-response dynamics.

On a slightly related note, papers such as~\cite{BarabasiAlbert1999, Watts2001,JacksonWatts2002,Snijders2001,Snijders2010} model network formation using random processes, without the context of a game, and analyze properties of network topologies that arise from such processes. 

\subsubsection*{Contributions}

\begin{enumerate}
	\item Unlike prior works~\cite{cisneros2021network, 2012_RT_CG,Salonen2016, Bei2011, LAOUTARIS20141266}, our model incorporates topological constraints, referred to as the \emph{underlying topology}, on network formation. Furthermore, unlike~\cite{cisneros2021network,Catalano2024}, we impose a hard budget constraint on the total weight that each agent can allocate across its outgoing links. While budget constraints have also been considered in~\cite{Salonen2016, Bei2011, LAOUTARIS20141266}, our model allows agents to have heterogeneous resource budgets, as in~\cite{Bei2011}.

	\item Unlike~\cite{Bei2011}, which considers betweenness centrality, we use Katz centrality to capture influence through both direct and indirect connections. This makes it particularly suitable for resource-allocation settings, such as investments and communication, where resources can enhance an agent's network-wide reach.

	\item Unlike the closely related work~\cite{2012_RT_CG}, we consider the formation of weighted networks. Moreover, we characterize stable network formations using Nash equilibrium rather than pairwise stability, thereby allowing unilateral link formation decisions in our directed network setting, without requiring mutual consent/benefit between agents. Additionally, our model allows more general directed graphs to arise as Nash equilibria, rather than being restricted to the complete and star graphs characterized as pairwise stable in~\cite{2012_RT_CG}.
	
	\item We show the \emph{mutual reinforcement property} of the Katz centrality. This property results in all Nash equilibrium networks having the same agent centrality vectors. We also show that there are Nash equilibrium networks that are sparse, that is where every agent has exactly one outgoing link. Under complete underlying topologies, we show that Nash equilibrium centralities are proportional to the agent's budgets. When the underlying topology allows for self-loops, we show that a hierarchical structure emerges in the condensation graph of Nash equilibria.

	\item Finally, we propose a sequential BRD for network formation and show that it converges to the set of Nash equilibria. This convergence is facilitated by the mutual reinforcement property and the fact that best responses yield networks in which agents exhaust their bounded budgets on neighbors with the highest centrality.
	\remend
\end{enumerate}

\subsubsection*{Notation}
Throughout the paper, we use non-bold letters for denoting scalars,
bold lowercase letters for denoting vectors, and bold uppercase
letters for denoting matrices. The sets of natural numbers, real
numbers, non-negative real numbers and positive real numbers are
denoted by $\mathbb{N},\real$, $\real_{\geq0}$ and $\real_{>0}$,
respectively. For any vectors $\boldw \in \realn$ and $\mathbf{x}_i \in \realn$, $w_j$ and $x_{ij}$ denote their $j$th elements, respectively and $\supp( \mathbf{x}_i ) := \{ j \in \{1, \ldots, n\}  \mid  x_{ij} \neq 0 \}$ is the support of the vector $\mathbf{x}_i$.
Let $\boldzero$ and $\boldone$ denote the vectors (of appropriate dimension) with all \emph{zero} and all \emph{one} elements, respectively. Let $\identityn$ denote the $n$ dimensional identity matrix. For any matrix $\adjmat \in \real^{n \times n}$, $\rho(\adjmat) \in \realnonneg$ denotes its spectral radius. For any vector $\boldw \in \realn$, $\boldw^\top$
denotes its transpose and $\diag(\boldw)$ denotes a diagonal matrix with $\boldw$ as its main diagonal. Let $\boldei$ denote the $i^{\mathsf{th}}$ standard basis vector of $\realn$. For any $a \in \real$, $|a|$ denotes its absolute value. For a collection of sets $\{ \Mc{S}_{i} \}_{i \in \{1, \dots, n\}}$, their Cartesian product is given by $\bigtimes_{i=1}^{n} \Mc{S}_{i}$. The empty set is denoted by $\varnothing$. We refer the reader to~\cite{2024_FB_network_book} for an introduction to basic graph theory and the standard definitions used throughout this paper.
\remend 

\section{Modeling and Problem Setup}\label{sec:prob}

Consider a set $\agt := \{1,\ldots,n\}$ of $n$ agents that are engaged in a network formation game. The agents seek to maximize their influence in the network by
allocating a limited resource to their social connections. We first introduce the resource constraints on the agents' choices. We then define an agent's \emph{Katz} centrality, which represents the agent's overall influence in the network. Finally, we describe the resulting network formation game and the process of network formation itself.

\subsubsection*{Resource allocation profiles, budget constraints, underlying topology and agents' Katz centrality}

We model the action or strategy of an agent in the network formation game as its allocation of a limited resource to the weights of its outgoing edges in the network. We denote the allocation profile of any agent $i \in \agt$ by $\boldwi := [ w_{i1}, \ldots, w_{in}]^\top \in \realn$, where the element $w_{ij}$ represents the resource (such as time, money, etc.) allocated by agent $i$ to agent $j$. Let $\boldw := [\boldw_1^\top, \ldots, \boldw_n^\top]^\top \in \real^{n^2}$ denote the allocation profile of all agents with $\boldwi$ as the $i\tth$ subvector. When we want to view the allocation profile from the perspective of agent $i \in \agt$, we write it as $\boldw = (\boldwi,\boldwnoti) \in \real^{n^2}$, where $\boldwi \in \real^{n}$ is the allocation profile of agent $i \in \agt$ and $\boldwnoti \in \real^{n (n-1)}$ is the allocation profile of all agents other than $i$.

The allocation profile $\boldw$ induces a weighted digraph $\grph(\boldw)= \grph = (\agt, \edg, \adjmat)$ with the adjacency matrix $\adjmat( \boldw ) := [\boldw_1, \ldots, \boldw_n]^\top \in \real^{n \times n}$. Agent $i \in \agt$ has a finite budget $B_i > 0$ for the total resources it can allocate. We further assume that the agents are constrained to allocate only to their social out-neighbors $\Ni(\grphunweighted)$ in an unweighted digraph, called the \textit{underlying topology} $\grphunweighted = (\agt, \edgunweighted, \adjmatunweighted)$. This models topological constraints on agents such as communication limitations, geographical proximity, or pre-existing social relationships. 
Formally, the resource allocation constraint set for any agent $i \in \agt$ is given by
\begin{align}
	&\constraintseti
	:= \left\{\boldwi \in \real^n_{\geq0} \mathrel{\bigg|} \supp(\boldwi) \subseteq \Ni (\grphunweighted)
	, \  
	\boldone^\top \boldwi
	\leq B_i  \right\}, \label{eq:agent_constraint_set}
\end{align}
where recall that $\sigma(\boldwi)$ is the support of $\boldwi$. Also, let $\constraintset := \bigtimes_{i=1}^{n} \constraintseti \subseteq \real^{n^{2}}_{\geq 0}$. We call $\boldwi \in \constraintseti$ a \emph{feasible} allocation of $i$. Note that, for any feasible allocation profile $\boldw \in \constraintset$, the unweighted version of the graph $\grph(\boldw)$ is a subgraph of $\grphunweighted$. The allocation profile $\boldw$ then determines the network centralities of the agents, which we introduce next. Henceforth, we omit the arguments of $\grph$, $\adjmat$ and $\Ni$ for brevity whenever no confusion arises. 

The \emph{Katz} centrality of any agent $i \in \agt$ in the network $\grph( \boldw )$ induced by allocation profile $\boldw = (\boldwi,\boldwnoti) \in \real^{n^2}$ is
\begin{equation}
\label{eq:katz_score}
\katzscorei(\boldw) = \katzscorei(\boldwi, \boldwnoti) := \inftysumk \agtsumj \delta^k \wijk,
\end{equation}
where $\delta \in (0, 1/\rho( \adjmat(\boldw) ) )$ is the discount factor and $\wijk$ is the $ij\tth$ element of $\adjmat^k(\boldw)$, where recall that $\adjmat( \boldw )$ is the adjacency matrix of the weighted graph
$\grph( \boldw )$. An agent $i$'s Katz centrality~\eqref{eq:katz_score} is the discounted sum of all (weighted) directed walks emanating from it in the network $\grph(\boldw)$.

Intuitively, if $\wij$ measures the \emph{direct} influence of $i$ on $j$, and $\wijk$ denotes the $k$-hop \emph{indirect} influence of $i$ on $j$, then the Katz centrality in~\eqref{eq:katz_score} quantifies the total influence of $i$ on all agents in $\agt$ by aggregating both direct and indirect influences. Here, the indirect influence is propagated through walks of arbitrary length, with longer walks discounted according to the discount factor $\delta$. Moreover, $\delta \in (0, 1/\rho( \adjmat ) )$ ensures that the infinite series in~\eqref{eq:katz_score} converges. In our setting, $\wij\geq 0$ represents the resources that agent $i$ allocates toward agent $j$, which may correspond to investment, advertising, social influence, communication effort, or other relationship-building activities. The budget $B_i>0$ in~\eqref{eq:agent_constraint_set} captures the limited resources available to agent $i$, such as financial resources, time, effort, attention, or computational capacity. The relevance of Katz centrality as a measure of network influence is supported by~\cite{2019_KB_SE_RI}, which shows that equilibrium production quantities, and consequently firm profits, in Cournot competition over networks are closely related to Katz centrality. Similarly,~\cite{ballester2006s} establishes a connection between Nash equilibrium actions and Katz centrality in network games. These connections motivate our use of Katz centrality as the utility measure in our model.

We impose the following standing assumptions on the underlying topology $\grphunweighted$ and the discount factor $\delta$, and justify it in the discussion that follows.

\begin{enumerate}[label=\textbf{(SA\arabic*)},wide=\parindent] 
\item $\grphunweighted$ is such that $\Ni(\grphunweighted) \neq \varnothing, \forall i \in \agt$.
\label{stand_asmp:weak_conn_of_underlying_graph}
\remend
\end{enumerate}	

\begin{enumerate}[resume,label=\textbf{(SA\arabic*)},wide=\parindent] 
\item $B_i < 1, \forall i \in \agt$ and $\delta = 1$ in~\eqref{eq:katz_score}.
\label{stand_asmp:budget_and_discount_factor}
\remend
\end{enumerate}	

Regarding Assumption~\ref{stand_asmp:weak_conn_of_underlying_graph}, note that if $\Ni(\grphunweighted)=\varnothing$, then $\constraintseti=\{\boldzero\}$ i.e., the only feasible allocation for agent $i$ is $\boldwi=\boldzero$, which consequently yields $\katzscorei(\boldw)=0, \forall \boldwnoti \in \constraintsetnoti:= \bigtimes_{j=1,j\neq i}^n \constraintset_j$. Thus, the allocations of other agents do not affect the centrality of agent $i$. Further, agent $i$ has no influence on the centralities of other agents, as implied by Lemma~\ref{lem:interdependence_betwn_katz_scores}. So, such an agent can be excluded from the network formation process.
We next justify Assumption~\ref{stand_asmp:budget_and_discount_factor}. Let $B_i>0$ be the budget of agent $i$. Given any allocation profile $\boldw \in \constraintset$, consider the allocation profile $\overline{\boldw}:=\delta\boldw$, with $(1/\delta)>\max_{i\in\agt} B_i$. Thus, $\adjmat(\overline{\boldw})=\delta \adjmat(\boldw)$. Then, letting $\katzscorestacked(\boldw):=[c_1(\boldw),\ldots,c_n(\boldw)]^\top$ denote the stacked vector containing the centralities of all agents, note that  
\begin{equation}\label{eq:vector-katz}
\katzscorestacked(\boldw)
=\inftysumk \delta^k\adjmat^k(\boldw)\boldone
=\inftysumk \adjmat^k(\overline{\boldw})\boldone.
\end{equation}
Further, the allocation profile $\overline{\boldw}=\delta \boldw$, which induces $\overline{\adjmat}$, satisfies the constraints in~\eqref{eq:agent_constraint_set} with resource budget parameters $\overline{B}_i=\delta B_i<1$. This equivalence justifies Assumption~\ref{stand_asmp:budget_and_discount_factor}.

Viewing centralities of all agents together, as in~\eqref{eq:vector-katz}, immediately yields the following result, which is useful in the subsequent analysis.
\begin{lemma}
\thmtitle{Interdependence between agent centralities}
\label{lem:interdependence_betwn_katz_scores}
Consider the Katz centralities defined in~\eqref{eq:katz_score}. Let $\boldw \in \constraintset$ be a feasible allocation profile, satisfying~\eqref{eq:agent_constraint_set}. Then, the Katz centrality of any agent $i \in \agt$ in $\grph( \boldw )$ satisfies
\begin{equation}
	\label{eq:interdependence_betwn_katz_scores}
	\katzscorei(\boldw) = \agtsumjniarg{\grphunweighted}\wij\left[1 + \katzscorej(\boldw)\right].
\end{equation}	
\end{lemma}

\begin{proof}
Under Assumption~\ref{stand_asmp:budget_and_discount_factor}, $\adjmat = \adjmat(\boldw)$ is strictly row sub-stochastic with $\rho(\adjmat(\boldw))<1, \forall \boldw \in \constraintset$. This implies that the stacked centrality vector $\katzscorestacked(\boldw)= \inftysumk \adjmat^k\boldone = \left(\identityn - \adjmat\right)^{-1}\adjmat \boldone$. Thus, $ \katzscorestacked(\boldw) = \adjmat [ \boldone + \katzscorestacked(\boldw) ]$. The claim now follows from the element-wise equalities.	
\end{proof}

\vspace{5pt}

\subsubsection*{Network formation game}

In this paper, we consider the \emph{strategic form game}
$\ntwgame(\grphunweighted) :=\langle \agt, \left(\constraintseti\right)_{i \in \agt}, \left(\katzscorei\right)_{i \in \agt}\rangle$ among the set of agents $\agt$, with the \emph{strategy} of agent
$i \in \agt$ being its allocation $\boldwi \in \constraintseti$ and its utility function being its Katz centrality $\katzscorei$. We refer to $\Mc{G}$ as the \emph{network formation game}. The set of \emph{Nash equilibria} of this game $\mathcal{G}$ is
\begin{align}
\notag \neset & := \{ \boldwstar \in \constraintset \,\,|\,\, \forall i \in \agt, \\
&\quad \katzscorei(\boldwistar,\boldwnotistar) \geq \katzscorei(\boldwi,\boldwnotistar), \forall \boldwi \in \constraintseti \}\,. \label{eq:neset}
\end{align}
Thus, $\boldwstar \in \neset$ if and only if $\boldwistar \in \Mc{BR}_{i}( \boldwnotistar )$ for all $i \in \agt$, where
$\displaystyle \bestrespseti := \argmax_{\boldwi \in \constraintseti} \katzscorei(\boldwi,\boldwnoti)$ is the set of best responses of agent $i$ to $\boldwnoti$. For a $\boldwstar \in \neset$, we refer to the graph
$\grph(\boldwstar)$, or $\negrph$ for short when there is no confusion, as a Nash equilibrium network.

\subsubsection*{Network formation process}

We also study the process of the network formation itself, which we model as the \emph{sequential} BRD. The process starts with an initial allocation profile $\boldw(0) \in \constraintset$. At each time step $k \in \nat$, an agent $i_k \in \agt$ is selected arbitrarily (randomly or otherwise) and the agent $i_k$ updates its allocation $\boldwik(k)$ to a best response to the allocation of the other agents $\boldwnotik(k-1)$, i.e.,
\begin{equation}
\label{eq:BRD}
\boldwik(k) \in \argmax_{\boldx \in \constraintsetik}\: \katzscoreik(\boldx,\boldwnotik(k-1)).
\end{equation}

We now briefly outline the main objectives of this work.
\subsubsection*{Objectives}
For the proposed network formation game, our goal is to understand the relationship between Katz centralities and resources within a Nash equilibrium network. Furthermore, we seek to characterize the set of all Nash equilibrium networks and analyze the convergence of BRD to this set. Finally, we aim to identify the structural properties of these equilibrium networks—specifically their sparsity and hierarchy—for various underlying network topologies.

\section{Analysis of the Game and Sequential BRD}

In this section, we analyze the network formation game $\ntwgame$ in detail. In particular, we investigate the existence and structural properties of Nash equilibrium networks of the game. We also analyze the convergence of BRD. We begin our analysis by establishing that, for any agent $i \in \agt$, a best response to any allocation
$\boldwnoti$ of the other agents always exists. This ensures that the BRD in~\eqref{eq:BRD} is well-posed.

We first provide an alternative representation of the Katz centralities~\eqref{eq:katz_score}, which will be useful in the subsequent analysis. Recall that the Katz centrality~\eqref{eq:katz_score} of any agent $i \in \agt$ in a graph $\grph$ induced by an allocation profile $\boldw = (\boldwi,\boldwnoti) \in \constraintset$ is given by the sum of the weights of all directed walks of all possible lengths emanating from $i$, where the weight of each walk is computed multiplicatively. For any $j \in \agt$ and any $i \in \agt\setminus\{j\}$, let $\mwalksumjtonoti(\boldw)$ be the sum of all weighted directed walks of length $m$ starting from node $j$ that do not reach node $i$, and $\mwalksumjtoi(\boldw)$ be the sum of all
weighted directed walks of length $m$ starting from node $j$ that reach node $i$ exactly once and terminate at $i$. Further, we define
\begin{align*}
&\pjiboldwnoti := \inftysumm \mwalksumjtonoti(\boldw), \ \ \qjiboldwnoti := \inftysumm \mwalksumjtoi(\boldw)
\\
&d_{ji}( \boldwnoti ) := \pjiboldwnoti + \qjiboldwnoti + 1,
\\
&q_{ii}( \boldwnoti ) := 1, \ \ d_{ii}( \boldwnoti ) := 1, \quad \forall i \in \agt.
\end{align*}
This along with Lemma~\ref{lem:interdependence_betwn_katz_scores} gives us the following alternate representation of~\eqref{eq:vector-katz}.
\begin{lemma}
\thmtitle{Alternate representation of centralities}
Consider an agent $j \in \agt$. For any $i \in \agtnotj:=\agt\setminus\{j\}$, the Katz centrality~\eqref{eq:katz_score} of $j$ can be re-expressed as
\begin{align}
	\label{eq:alternate_formulation_of_katz_score}
	\notag \katzscorej(\boldw) &= \inftysumm \mwalksumjtonoti(\boldw) + \inftysumm \mwalksumjtoi(\boldw)\left[1 +  \katzscorei(\boldw)\right], \\
	&=: \pjiboldwnoti + \qjiboldwnoti[1+\katzscorei(\boldw)]. 
\end{align} 
\end{lemma}
\begin{proof}
For any $\boldw \in \constraintset$, from Lemma~\ref{lem:interdependence_betwn_katz_scores}, recall that
\begin{equation*}
	\katzscorej(\boldw)=\agtsumk \wjk+ \wji\katzscorei(\boldw) +\agtsumknoti\wjk\katzscorek(\boldw).
\end{equation*}
Substituting for $\katzscorek(\boldw)$ using~\eqref{eq:interdependence_betwn_katz_scores} in the above gives
\begin{align*}
	\katzscorej(\boldw) &=  \agtsumknoti\wjk+\agtsumknoti\agtsumlnoti\wjk\wkl[1+\katzscorel(\boldw)] \\ 
	&\quad + [\wji+\agtsumknoti\wjk\wki]\left[1+\katzscorei(\boldw)\right].
\end{align*}
The claim now follows by recursively substituting for $\katzscorel(\boldw)$ using~\eqref{eq:interdependence_betwn_katz_scores} on the right-hand side and rearranging terms.
\end{proof}

Note that, since $\mwalksumjtonoti(\boldw)$ only includes walks that do not pass through $i$ and $\mwalksumjtoi(\boldw)$ only
includes walks that reach $i$ exactly once and terminate at $i$, $\pjiboldwnoti$, $\qjiboldwnoti$ and $d_{ji}( \boldwnoti )$ are independent of $\boldwi$ and depend only on $\boldwnoti$. Then, for any $\boldw=(\boldwi,\boldwnoti) \in \constraintset$, from~\eqref{eq:interdependence_betwn_katz_scores},~\eqref{eq:alternate_formulation_of_katz_score}, we have
\begin{equation}\label{eq:katz_score_fractional_linear}
\katzscorei(\boldwi, \boldwnoti)
= \frac{\displaystyle \sum_{j \in \Ni} d_{ji}( \boldwnoti ) \wij }{\displaystyle 1 - \sum_{j \in \Ni} \qjiboldwnoti \wij},\: \forall i \in \agt.
\end{equation}.

The following lemma gives a basic fact that we reuse later.
\begin{lemma}
\label{lem:positiv-denom}
Let $\boldw \in \constraintset$. Then,
\begin{equation*}
	1 - \sum_{{j \in \Ni}} \qjiboldwnoti \wij > 0, \quad \forall i \in \agt.
\end{equation*}
\end{lemma}
\begin{proof}
Consider any $\boldw \in \constraintset$ and any $i \in \agt$. If $\boldwi = \boldzero$ then the claim holds trivially. If $\boldwi \neq \boldzero$ then from~\eqref{eq:katz_score},
$\katzscorei(\boldwi, \boldwnoti) > 0, \forall \boldwnoti \in \constraintsetnoti$. Since $\pjiboldwnoti \geq 0$ and $\qjiboldwnoti \geq 0$,
$\pjiboldwnoti + \qjiboldwnoti + 1 > 0, \forall j \in \Ninoti$. Hence, $\sum_{j \in \Ni} d_{ji}( \boldwnoti ) \wij>0$. The claim now follows from~\eqref{eq:katz_score_fractional_linear}.
\end{proof}

We are now ready to state the next result, which establishes (among other things) that a best response allocation $\boldwi \in \constraintseti$ to $\boldwnoti \in \constraintsetnoti$ always exists for any agent $i \in \agt$.

\begin{lemma}
\thmtitle{On best response set}
\label{lem:existence_of_best_response}
Let $\grphunweighted$ be a given underlying topology.
Consider the network formation game $\Mc{G}$ and any $\boldw \in \constraintset$. 
For any agent
$i \in \agt$, 
$\bestrespseti$ is non-empty, convex and
\begin{equation*}
	B_i\boldej \in \bestrespseti, \ \forall j \in \argmax_{k \in \Ni(\grphunweighted)} \left\{ \frac{d_{ki}( \boldwnoti )}{1 - q_{ki}( \boldwnoti ) B_{i}} \right\}.
\end{equation*}
\end{lemma}

\begin{proof}
Since $\boldwnoti \in \constraintsetnoti$ is a fixed parameter as far as $\bestrespseti$ is concerned, we will drop $\boldwnoti$ from most of the notation in this proof.

\emph{Non-emptiness and single edge allocations in $\bestrespseti$:}  
We first introduce a change of variables and rewrite the expression for $\katzscorei(\boldwi ,\boldwnoti)$ in~\eqref{eq:katz_score_fractional_linear} more concisely. Let
\begin{align*}
	&f_{ji} := \frac{d_{ji}}{1 - q_{ji} B_{i}}, \ z_{ij} := \frac{(1 - q_{ji} B_i) w_{ij}}{1 - \sum_{k \in \Ni} q_{ki} w_{ik}},  \quad \forall j \in \agt.
\end{align*}
Now, notice from~\eqref{eq:katz_score_fractional_linear} that $\katzscorei(\boldwi, \boldwnoti) = \sum_{j \in \Ni} f_{ji} z_{ij}$. Moreover, Lemma~\ref{lem:positiv-denom} implies that $( 1 - q_{ji} B_{i} ) > 0, \forall j \in \Ni$. Further,
\begin{equation*}
	\sum_{j \in \Ni} z_{ij} = \frac{\sum_{j \in \Ni} w_{ij} - B_{i} \sum_{j \in \Ni} q_{ji} w_{ij}}{1 - \sum_{j \in \Ni} q_{ji} w_{ij}} ,
\end{equation*}
from which we can reason that $\sum_{j \in \Ni} z_{ij} \leq B_i$ iff $\sum_{j \in \Ni} w_{ij} \leq B_i$. From the above definitions, observations, and Lemma~\ref{lem:positiv-denom}, we can say that
\begin{enumerate}[label=\textbf{(S\arabic*)}] 
	\item If $\boldwi \in \constraintseti$ then $\boldzi \in \constraintseti$ and $\supp( \boldwi ) = \supp( \boldzi )$.
\end{enumerate}
\begin{enumerate}[resume,label=\textbf{(S\arabic*)}]
	\item If $\boldzi \in \constraintseti$ then $\boldwi \in \constraintseti$  and $\supp( \boldwi ) = \supp( \boldzi )$.
\end{enumerate}
Next, from the definition of $z_{ij}$, notice that 
\begin{align*}
	& \frac{1}{1 - \sum_{j \in \Ni} q_{ji} w_{ij}} - \sum_{j \in \Ni} \frac{q_{ji} z_{ij}}{1 - q_{ji} B_i}
	\\
	&= \frac{1}{1 - \sum_{j \in \Ni} q_{ji} w_{ij}} \left[ 1 - \sum_{j \in \Ni} q_{ji} w_{ij} \right] = 1 .
\end{align*}
Thus,
\begin{equation*}
	1 - \sum_{k\in \Ni} q_{ki} w_{ik} = \frac{1}{1 + \sum_{k \in \Ni} \frac{q_{ki} z_{ik}}{1 - q_{ki} B_i} } .
\end{equation*}  
Using this together with the definition of $z_{ij}$, we obtain the inverse map from $\boldzi$ to $\boldwi$ as follows

\begin{equation*}
	w_{ij} = \frac{z_{ij}}{(1 - q_{ji}B_i)(1 + \sum_{k \in \Ni} \frac{q_{ki} z_{ik}}{1 - q_{ki}B_i})}.
\end{equation*}
Thus, the task of finding the best response set $\bestrespseti$ is equivalent to solving and transforming the set
\begin{equation}
	\Mc{S} := \argmax_{\boldzi \in \constraintseti} \sum_{j \in \Ni} f_{ji} z_{ij}
	\label{eq:LP-after-transform}
\end{equation}
back to the $\boldw$ space. The optimization problem in~\eqref{eq:LP-after-transform} is a linear program and given $\constraintseti$, we can say that $B_i \boldej \in \Mc{S}, \forall j \in \argmax_{j \in \Ni} f_{ji}$. Notice from the inverse map that if $\boldzi = B_i \boldej$ then $\boldwi = \boldzi = B_i \boldej$. Hence, $B_i \boldej \in \bestrespseti, \forall j \in \argmax_{j \in \Ni} f_{ji}$. This proves the non-emptiness of $\bestrespseti$ and the existence of single edge allocations in $\bestrespseti$.

\emph{Convexity of $\bestrespseti$:} From~\eqref{eq:katz_score_fractional_linear}, we can write $\katzscorei(\boldwi, \boldwnoti)
= \displaystyle \frac{ \bolddi^\top \boldwi }{1 - \boldqi^{\top} \boldwi }$. Let $\hat{\boldw}_i, \bar{\boldw}_i \in \bestrespseti$ be two best responses by agent $i$ to $\boldwnoti$. We thus have
\begin{equation*}
	\bolddi^\top \bar{\boldw}_i = \bolddi^\top \hat{\boldw}_i \left( \frac{ 1 - \boldqi^{\top} \bar{\boldw}_i }{1 - \boldqi^{\top} \hat{\boldw}_i } \right).
\end{equation*}
Now consider $\boldwi := \lambda \hat{ \boldw }_i + (1 - \lambda) \bar{  \boldw }_i$, with $\lambda \in [0, 1]$. So,
\begin{align*}
	\katzscorei(\boldwi, \boldwnoti) %
	&= \frac{ \bolddi^\top [\lambda \hat{ \boldw }_i + (1 - \lambda) \bar{  \boldw }_i] }{1 - \boldqi^{\top} [\lambda \hat{ \boldw }_i + (1 - \lambda) \bar{  \boldw }_i] }
	\\
	&= \frac{\bolddi^\top \hat{\boldw}_i \left[ \lambda + (1 - \lambda)\left( \frac{ 1 - \boldqi^{\top} \bar{\boldw}_i }{1 - \boldqi^{\top} \hat{\boldw}_i } \right) \right]}{1 - \boldqi^{\top} [\lambda \hat{ \boldw }_i + (1 - \lambda) \bar{  \boldw }_i] }
	\\
	&= \katzscorei( \hat{\boldwi}, \boldwnoti).
\end{align*}
So, $\boldwi \in \bestrespseti$, and hence $\bestrespseti$ is convex.
\end{proof}

Lemma~\ref{lem:existence_of_best_response} shows the non-emptiness and convexity of the best response set for any agent $i \in \agt$. Further there exist some best responses to $\boldwnoti$ wherein the entire resource budget is allocated to a single well-chosen agent $j \in \Ni$.

The next result shows certain properties of \emph{better responses}. Recall that a better response of any agent $i \in \agt$ to the allocations $\boldwnoti \in \constraintsetnoti$ of the other agents with respect to $\boldwi \in \constraintseti$ is any allocation $\boldyi \in \constraintseti$ such that $\katzscorei(\boldyi,\boldwnoti) \geq \katzscorei(\boldwi,\boldwnoti)$. If the above inequality is strict, then $\boldyi$ is called a \emph{strict better response}. We will use $\betterrespseti$ and $\strictbetterrespseti$ to denote the sets of better responses and strict better responses with respect to $\boldwi \in \constraintseti$, respectively. Note that $\bestrespseti \subseteq \betterrespseti$ and $\strictbetterrespseti \subseteq \betterrespseti$. Moreover, whenever $\strictbetterrespseti \neq \varnothing$, it holds that $\bestrespseti \subseteq \strictbetterrespseti$.

\begin{lemma}
\thmtitle{On the better and best responses}
\label{lem:properties_of_better_resp}
Let $\grphunweighted$ be a given underlying topology. Consider the network formation game $\Mc{G}$ with Katz centralities given in~\eqref{eq:katz_score}. Define the function \begin{equation}
\label{eq:v_function}
\Vi(\boldx):=B_i[1+\max_{j \in \Ni}x_j], \ \boldx \in \realn_{\geq 0}.
\end{equation} Then for any agent $i \in \agt$ and $\boldw = (\boldwi,\boldwnoti)\in \constraintset$,
\begin{enumerate}
\item \label{lem:btr_resp_1} If $\boldyi \in \betterrespseti$ then $ \katzscorej(\boldyi,\boldwnoti) \geq \katzscorej(\boldwi,\boldwnoti),$ $\forall j \in \agt$. 

\item \label{lem:btr_resp_2} $\boldyi \in \strictbetterrespseti$ if and only if $\agtsumjni (\yij - \wij)[1+\katzscorej(\boldw)]>0$.

\item \label{lem:btr_resp_3} $\boldwi \in \bestrespseti$ if and only if $\strictbetterrespseti = \varnothing$ if and only if
$\katzscorei(\boldw)=\Vi(\katzscorestacked(\boldw))$.
\end{enumerate}
Moreover,
	for any $\boldwi \in \bestrespseti$, $\agtsumjni \wij=B_i$ and $\wij>0$ implies $j\in\argmax_{k\in\Ni} \katzscorek( \boldwi, \boldwnoti )$.
\end{lemma}

\begin{proof}
	Consider an agent $i \in \agt$ and an allocation profile $\boldw \in \constraintset$. Recall the form of $\katzscorej(\boldwi,\boldwnoti)$ from~\eqref{eq:alternate_formulation_of_katz_score}. Suppose $\boldyi \in \betterrespseti$. Claim~\ref{lem:btr_resp_1} now follows from the fact that $\qjiboldwnoti \geq 0$ and $\pjiboldwnoti \geq 0$, and that both are independent of $\boldwi$. 
	
	Next, for any $\boldyi \in \constraintseti$, we can use~\eqref{eq:interdependence_betwn_katz_scores} and~\eqref{eq:alternate_formulation_of_katz_score} to obtain
	\begin{equation*}
		\katzscorei(\boldyi,\boldwnoti) - \katzscorei(\boldwi,\boldwnoti) = \frac{\agtsumjni (\yij - \wij)[1+\katzscorej(\boldw)]}{\displaystyle 1 -  \sum_{j \in \Ni} \yij  \qjiboldwnoti}.
	\end{equation*}
	Claim~\ref{lem:btr_resp_2} now follows from Lemma~\ref{lem:positiv-denom}.
	
	Next, notice by definition, $\boldwi \in \bestrespseti$ if and only if $\strictbetterrespseti = \varnothing$. From Claim~\ref{lem:btr_resp_2}, observe that
	\begin{equation*}
		\max_{\boldyi \in \constraintseti} \agtsumjni \yij[1+\katzscorej(\boldw)] = \agtsumjni \wij[1+\katzscorej(\boldw)] = \katzscorei(\boldw),
	\end{equation*}
	where the last equality follows from~\eqref{eq:interdependence_betwn_katz_scores}. Using the definition of $\constraintseti$ given in~\eqref{eq:agent_constraint_set}, and the fact that $\katzscorei(\boldw)\geq 0, \forall \boldw \in \constraintset$, it follows that  $\max_{\boldyi \in \constraintseti} \agtsumjni \yij[1+\katzscorej(\boldw)] = B_i[1+\max_{j \in \Ni}\katzscorej(\boldw)]$. Claim~\ref{lem:btr_resp_3} now follows from~\eqref{eq:v_function}. 
	
	Finally, from ~\eqref{eq:v_function}, ~\eqref{eq:interdependence_betwn_katz_scores} and Claim 3, it can be easily seen that for any $\boldwi \in \bestrespseti$, the following equality holds
	\begin{equation*}
			\agtsumjni \wij[1+\katzscorej(\boldwi,\boldwnoti)] = B_i(1+\max_{j \in \Ni}\katzscorej(\boldwi,\boldwnoti)).
	\end{equation*}
	Since $\boldwi \in \constraintseti$, the above equality holds if and only if $\agtsumjni \wij = B_i$ and $\wij \ge 0$. Further, $\wij > 0$ implies $j \in \argmax_{k \in \Ni}\katzscorek(\boldwi,\boldwnoti)$.
	The proof is now complete.
\end{proof}
	
Lemma~\ref{lem:properties_of_better_resp} provides the following insights. At a given allocation profile $\boldw \in \constraintset$, if any individual agent chooses a better response allocation $\boldyi \in \betterrespseti$, then it does not decrease the Katz centralities of the other agents. The same holds for best response allocations $\boldyi \in \bestrespseti$, since a best response is also a better response.
The Lemma also implies that an agent $i \in \agt$'s best response set to $\boldwnoti$ is the same as the agent's best response to the current centralities $\katzscorestacked(\boldw)$. This is quite remarkable and in fact  helps us provide a necessary and sufficient condition for a network $\grph(\boldw)$ induced by an allocation profile $\boldw \in \constraintset$ to be a Nash equilibrium network. 

\begin{theorem}
	\thmtitle{All Nash equilibrium networks share the same centrality profile}
	\label{thm:nash_set_characterization}
	Let $\grphunweighted$ be a given underlying topology. Consider the network formation game $\Mc{G}$ with Katz centralities given in~\eqref{eq:katz_score}. Let $\Vi(\cdot)$ be as defined
	in~\eqref{eq:v_function}. Then,
	\begin{enumerate}
		\item \label{thm:nash_eqm_set_1}
		$\boldw^* \in \neset$ if and only if $\Vi(\katzscorestacked(\boldw^*)) = \katzscorei(\boldw^*), \forall i \in \agt$.
		
		\item \label{thm:reln_betwn_nash_eqm_centralities}
		For any agent $i\in \agt$, if $\wij^* > 0$ for some $j \in \Ni(\grphunweighted)$ then $\katzscorei(\boldw^*)=B_i[1+\katzscorej(\boldw^*)]$.
		
		\item \label{thm:nash_eqm_set_2}
		$\exists$ an unique $\katzscorestacked^* \in \realn_{\geq 0}$  such that $\neset = \{\boldw \in \constraintset \mid \katzscorestacked(\boldw) = \katzscorestacked^*\}$.
	\end{enumerate}
\end{theorem}

\begin{proof}
	By definition $\boldw^*=(\boldwi^*,\boldwnoti^*) \in \neset$ if and only if $\boldwi^* \in \Mc{BR}_i(\boldwnoti^*), \forall i \in \agt$. Claim~\ref{thm:nash_eqm_set_1} now follows from Lemma~\ref{lem:properties_of_better_resp}. If $\wij^* >0, j \in \Ni(\grphunweighted)$, then from Lemma~\ref{lem:properties_of_better_resp}, $j \in \argmax_{k \in \Ni(\grphunweighted)} \katzscorek(\boldw^*)$. Claim~\ref{thm:reln_betwn_nash_eqm_centralities} follows from~\eqref{eq:v_function}.
	
	Finally, let $\boldv(\boldx):=[\V_1(\boldx), \ldots, \V_n(\boldx)]^\top$, where $\Vi(\boldx)$ is as defined in~\eqref{eq:v_function}.  For any
	$\boldw \in \constraintset$, we know from~\eqref{eq:katz_score} that $\katzscorestacked(\boldw)\geq \boldzero$. Thus, $\boldv(\cdot): \realn_{\geq 0} \to \realn_{\geq 0}$. Consider any
	$\boldx,\boldy \in \realn_{\geq 0}$. Then
	\begin{align*}
		\|\boldv(\boldx)-\boldv(\boldy)\|_{\infty} 
		&=   \max_{i \in \agt}B_i\left\{\:\left| \max_{j \in \Ni}x_j - \max_{j \in \Ni}y_j\right|\: \right\}, \\
		&\leq \max_{i \in \agt} B_i \|\boldx - \boldy\|_{\infty}.
	\end{align*}
	The above implies that $\boldv(\cdot)$ is a contraction on $\realn_{\geq 0}$, by virtue of~\ref{stand_asmp:budget_and_discount_factor}. Since ($\realn_{\geq 0}$, $\|\cdot\|_{\infty}$) is a complete metric space, by the Banach fixed point theorem, $\boldv(\cdot)$ has a unique fixed point $\katzscorestacked^* \in \realn_{\geq 0}$. From Claim~\ref{thm:nash_eqm_set_1}, we know that $\boldw^* \in \neset$ if and only if $\boldv(\katzscorestacked(\boldw^*))=\katzscorestacked(\boldw^*)$. Claim~\ref{thm:nash_eqm_set_2} now follows.
\end{proof}

Claim~\ref{thm:nash_eqm_set_1} in Theorem~\ref{thm:nash_set_characterization} states that the network $\grph(\boldw^*)$ induced by $\boldw^* \in \constraintset$ is a Nash equilibrium network if and only if the resulting centrality vector $\katzscorestacked(\boldw^*)$ is a fixed point of the function $\boldv(\cdot)$. Claim~\ref{thm:nash_eqm_set_2} in Theorem~\ref{thm:nash_set_characterization} further states that any two distinct Nash equilibrium networks (if they exist) yield the same Katz centralities $\katzscorestacked^*$. One can use the Banach fixed-point iteration to compute $\katzscorestacked^*$. 

The next result establishes an invariance property of $\neset$ under BRD for any general underlying topology $\grphunweighted$. Specifically, it shows that if an agent unilaterally switches to another best-response strategy, then the resulting network is still a Nash equilibrium network. Thus, no other agent has an incentive to deviate, as their centralities cannot be improved.
\begin{corollary}
	\thmtitle{Invariance of $\neset$ under unilateral best response deviations}
	Let $\grphunweighted$ be a given underlying topology. Consider the network formation game $\Mc{G}$ with Katz centralities given in~\eqref{eq:katz_score}. Let $\boldw^* =(\boldwi^*,\boldwnoti^*) \in \neset$. Suppose $\exists i \in \agt$ and $\boldx =(\boldxi,\boldw_{-i}^*) \in \constraintset$ such that $\katzscorei(\boldw^*) = \katzscorei(\boldx)$. Then, $\boldx \in \neset$.
\end{corollary}
\begin{proof}
	Under the stated assumptions, $\boldxi \in \Mc{BR}_{i}(\boldwnoti^*)$. Also, $\boldwi^* \in \Mc{BR}_i(\boldwnoti^*)$. Thus, from~\eqref{eq:alternate_formulation_of_katz_score}, we have that $\katzscorej(\boldw^*) = \katzscorej(\boldx ) ,\forall j \in \agt$. Hence, $\katzscorestacked(\boldx ) = \katzscorestacked(\boldw^*)$. The result now follows from Claim~\ref{thm:nash_eqm_set_2} in Theorem~\ref{thm:nash_set_characterization}.
\end{proof}

Finally, we conclude this section with the main convergence result for BRD. 
\begin{theorem}
	\thmtitle{Convergence of BRD to $\neset$}
	Let $\grphunweighted$ be a given underlying topology. For a given initial network $\boldw(0) \in \constraintset$, consider the sequential BRD~\eqref{eq:BRD} with an agent update sequence $\{i_k\}_{k \in \nat}$ in which every agent in $\agt$ updates infinitely often (i.o.). Let $\boldw(k)$ denote the network at time step $k$ generated by this BRD. Then, the sequence of networks $\{\boldw(k)\}_{k \in \nat}$ converge to $\neset$, i.e. all subsequential limits of $\{\boldw(k)\}_{k \in \agt}$ are in $\neset$.
\end{theorem}

\begin{proof}
	Consider any $\boldw(0) \in \constraintset$.
	Under the sequential BRD~\eqref{eq:BRD} with an agent update sequence $\{i_k\}_{k \in \nat}$ in which each agent appears i.o, Lemma~\ref{lem:properties_of_better_resp} implies that, for every $i \in \agt$, the sequence $\{\katzscorei(\boldw(k))\}_{k \in \nat}$, with $\boldw(k)=(\boldwi(k),\boldwnoti(k))$, is monotonically increasing. It is also bounded since $\katzscorei(\cdot)$ is continuous, $\boldw(k) \in \constraintset, \forall k$, and $\constraintset$ is compact. Hence, $\{\katzscorestacked(\boldw(k))\}_{k \in \nat}$ converges to some $\katzscorestacked^* \in \realn_{\geq 0}$. Now, observe that under~\eqref{eq:BRD}, at any time step $k \in \nat$, the agent $i_k$ chooses $\boldwik(k) \in \Mc{BR}_{i_k}(\boldwnotik(k-1))$, and hence, from Lemma~\ref{lem:properties_of_better_resp}, $v_{i_k}(\katzscorestacked(\boldw(k))) = \katzscoreik(\boldw(k))$. This implies that
	$\lim_{k \to \infty} \boldv(\katzscorestacked(\boldw(k))) 
	= \boldv\!\left(\lim_{k \to \infty}\katzscorestacked(\boldw(k))\right) 
	= \boldv(\katzscorestacked^*) 
	= \katzscorestacked^*,$
	where we have used the fact that $\Vi(\cdot),\forall i \in \agt$, as defined in Lemma~\ref{lem:properties_of_better_resp}, is continuous. The claim now follows from Theorem~\ref{thm:nash_set_characterization} and its proof.
\end{proof}
		
From Theorem~\ref{thm:nash_set_characterization} and its proof, it is clear that there is a unique fixed point $\katzscorestacked^*$ of the $\boldv(.)$ map and that $\katzscorestacked(\boldw) = \katzscorestacked^*, \forall \boldw \in \neset$. However, it is not obvious that there exists an admissible network with weights $\boldw$ such that $\katzscorestacked(\boldw) = \katzscorestacked^*$ or equivalently, it is not obvious that there exists a Nash equilbrium network.

\begin{remark}
	\thmtitle{On the existence of Nash equilibrium network} For any $\boldw \in \constraintset$, let $\Mc{V}_{s}(\boldw):= \{i \in \agt \mid \strictbetterrespseti \neq \varnothing \}$. We can then consider the following modified BRD. Given any initial network $\boldw(0) \in \constraintset$, at any time step $k \in \nat$, choose an agent $i_k \in \Mc{V}_s(\boldw(k-1))$ and restrict the best response to $\boldwik(k) = B_{i_k}\boldej \in \Mc{BR}_{i_{k}}(\boldwnotik(k-1))$, as established by Lemma~\ref{lem:existence_of_best_response}. One can then use the strictly monotone evolution of centralities under the modified BRD, along with the finiteness of possible networks due to the structure of the best responses above, to show convergence of this modified BRD to a Nash equilibrium in finite time. Hence, $\neset \neq \varnothing$. Since this modified BRD is not the focus of this paper, we omit the proof for brevity. \remend
\end{remark}

\section{About the Nash Equilibrium Networks}
Here, we analyze the properties of Nash equilibrium networks under specific underlying topologies. We first consider the case where the underlying topology $\grphunweighted$ is complete. 
\begin{theorem}
	\thmtitle{Properties of Nash equilibrium networks when $\grphunweighted$ is complete}
	\label{thm:nash_eqm_properties_complete_underlying_ntw}
	Consider the network formation game $\Mc{G}$ with Katz centralities given in~\eqref{eq:katz_score} and suppose that the underlying topology $\grphunweighted$ is a  complete digraph. Define $\Mc{H}:= \{\boldw \in \constraintset \mid \forall i \in \agt, \agtsumjniarg{\grphunweighted} \wij = B_i, \wij > 0 \Rightarrow B_j = B_M\}$
	where $B_M := \max_{i \in \agt} B_i$. Then $ \Mc{H} = \neset = \left\{\boldw \in \constraintset \mid
	\katzscorei(\boldw) = \tfrac{B_i}{1-B_M},\forall i \in  \agt\right\}$.
\end{theorem}

\begin{proof}
	As $\grphunweighted$ is complete, $\Ni(\grphunweighted)=\agt,\forall i \in \agt$. From~\eqref{eq:v_function} and Claim~\ref{thm:nash_eqm_set_1} in Theorem~\ref{thm:nash_set_characterization}, we know that $\boldw^* \in \neset$ if and only if $\Vi(\katzscorestacked(\boldw^*)) = \katzscorei(\boldw^*), \forall i \in \agt$ if and only if $\katzscorei(\boldw^*)=B_i(1+\max_{j \in \agt}\katzscorej(\boldw^*)), \forall i \in \agt$. Also, we know that $\exists \katzscorestacked^*$ such that $\katzscorestacked(\boldw^*)=\katzscorestacked^*, \forall \boldw^*\in \neset$. Consider $l \in \argmax_{i \in \agt} c_i^*$. Then, $c_l^* = B_l/(1 - B_l)$ and  $c_i^*=B_i/(1-B_l), \forall i \in \agt$. 
	This along with Theorem~\ref{thm:nash_set_characterization} proves that $\katzscorei(\boldw) = \tfrac{B_i}{1-B_M}$ for all $\boldw \in \neset$.

	In order to prove that $\Mc{H} = \neset$, let us define $\agt_M:= \{i \in \agt \mid B_i = B_M\}$. Let $\boldw \in \Mc{H}$ and let $\grph(\boldw)$ be the graph induced by it. It is easy to see that for any $i \in \agt_M$, $\wij > 0$ only if $j \in \agt_M$. Let $\adjmat_M$ denote the adjacency matrix of the subgraph of $\grph(\boldw)$ induced~\footnote{The term “induced” is used here in the standard graph-theoretic sense.} by $\agt_M$. Then, $\adjmat_M\boldone=B_M\boldone$ and $\adjmat_M$ is strictly row sub-stochastic by virtue of~\ref{stand_asmp:budget_and_discount_factor}. This means that $(\mathbf{I}-\adjmat_M)^{-1}\boldone - \boldone = \tfrac{B_M}{1-B_M}\boldone$ and from~\eqref{eq:katz_score}, $\katzscorei(\boldw) = \tfrac{B_M}{1-B_M},\forall i \in \agt_M$. Thus, for any $j \notin \agt_M$, from Lemma~\ref{lem:interdependence_betwn_katz_scores} we get $\katzscorej(\boldw)=\tfrac{B_j}{1-B_M}$. From the second equality of the claim, we get that $\boldw \in \neset$. Next, suppose $\boldw \in \neset$. Based on the second equality of this claim, it follows that, $\katzscorei(\boldw) = \frac{B_i}{1-B_M}, \forall i \in \agt$. Hence, from Lemma \ref{lem:properties_of_better_resp}, $\agtsumjniarg{\grphunweighted} \wij = B_i$ and for any $i,j \in \agt$ such that $\wij > 0$, we have $\katzscorej(\boldw) = \max_{k \in \agt} \katzscorek(\boldw) = \frac{B_M}{1-B_M}$. Thus, $B_j = B_M$. This implies $\boldw \in \Mc{H}$.
\end{proof}

From Theorem~\ref{thm:nash_eqm_properties_complete_underlying_ntw}, we see that when the underlying topology $\grph$ is complete, then at any Nash equilibrium network $\grph^*$ induced by $\boldw^* \in \neset$, for any two agents $i,j \in \agt$, $\katzscorei(\boldw^*) \leq \katzscorej(\boldw^*)$ if and only if $B_i \leq B_j$. Theorem~\ref{thm:nash_eqm_properties_complete_underlying_ntw} also shows that when the underlying topology is complete, every agent allocates its entire budget among the agents with the maximum budget in a Nash equilibrium network. For example, this means that in the case where there is only one agent with maximum resources, i.e., $|\agt_M|=1$, the Nash equilibrium network corresponds to a \emph{star} network, with every other agent allocating its entire budget to that particular agent, which consequently has the highest centrality.

Next, we study the properties of cycles in Nash equilibrium networks when $\grphunweighted$ is undirected. In particular, the following result characterizes the centralities of agents belonging to a directed cycle, if one exists, in such a Nash equilibrium network.
\begin{theorem}
	\thmtitle{Properties of directed cycles in the Nash equilibrium network when $\grphunweighted$ is undirected}
	Let $\grphunweighted$ be a given underlying topology and suppose that $\grphunweighted$ is undirected. Consider the network formation game $\Mc{G}$ with Katz centralities given in~\eqref{eq:katz_score}. Let $\boldw^* \in \neset$, and let $\grph^* := \grph(\boldw^*)$ denote the Nash equilibrium network induced by it. Consider any directed cycle (if there exists one) of length $l$ in $\grph^*$. Then, if $l$ is odd, all agents in the cycle have the same budgets and centralities. Alternatively, if $l$ is even, every pair of alternate agents in the cycle has the same budgets and centralities. \remend
\end{theorem}
\begin{proof}
	As $\grphunweighted$ is undirected, $j \in \Ni(\grphunweighted)$ if and only if $i \in \Ni(\grphunweighted)$. Let $\boldw^* \in \neset$ and let $\katzscorestacked^*$ be the centrality vector associated with it, as guaranteed by Theorem~\ref{thm:nash_set_characterization}. Consider a directed cycle $i_1 \mapsto i_2 \mapsto \ldots \mapsto i_r \mapsto i_1$ of length $r$ in $\grph^*$. For any $l \in \{2,\ldots,r-1\}$, Lemma~\ref{lem:properties_of_better_resp} gives $c_{i_{l+1}}^*\geq c_{i_{l-1}}^*$. Further, $c_{i_2}^*\geq c_{i_r}^*$ and $c_{i_1}^*\geq c_{i_{r-1}}^*$. The claim now follows by combining these inequalities, and reasoning depending on whether $r$ is even or odd.
\end{proof}
		
We now consider general underlying graph topologies, albeit with the following assumption.

\begin{enumerate}[label=\textbf{(A\arabic*)},wide=\parindent] 
	\item\thmtitle{All agents have self-loops in the underlying topology} In the digraph $\grphunweighted$, $i \in \Ni(\grphunweighted), \ \forall i \in \agt$.
	\remend
	\label{asmp:self_loops_in_underlying_ntw}
\end{enumerate}
\begin{lemma}
	\thmtitle{Agents connect only to nodes with centrality at least as high as their own in a Nash equilibrium network with self-loops allowed in $\grphunweighted$}
	\label{lem:neighbors_have_higher_katz_at_nash_eqm}
	Let $\grphunweighted$ be a given underlying topology and suppose that Assumption~\ref{asmp:self_loops_in_underlying_ntw} holds. Consider the network formation game $\Mc{G}$ with Katz centralities given in~\eqref{eq:katz_score}. Let $\boldw^* \in \neset$ and let $\grph^*:=\grph(\boldw^*)$ denote the Nash equilibrium network induced by it. Then for any $ i \in \agt$, if $j \in \Ni(\grph^*)$ then $\katzscorei(\boldw^*) \leq \katzscorej(\boldw^*)$.
\end{lemma}
\begin{proof}
	Since $\boldw^* \in \neset$, $\boldwi^* \in \Mc{BR}_{i}(\boldwnoti^*), \forall i \in \agt$, by definition. From Lemma~\ref{lem:properties_of_better_resp}, we know that if $j \in \Ni(\grph^*)$, then $j \in \argmax_{k \in \Ni(\grphunweighted)} \katzscorek(\boldw^*)$. The claim now follows since $i \in \Ni(\grphunweighted)$ under the stated assumptions. 
\end{proof}

In the following result, we show that agents in an SCC of Nash equilibrium network $\grph^*=\grph(\boldw^*)$ induced by $\boldw^* \in \neset$ have same budget and centralities. 
\begin{theorem}
	\thmtitle{Agents belonging to a SCC of a Nash equilibrium network with self-loops in $\grphunweighted$ have same budget and centralities}
	\label{thm:agents_in_scc_have_same_centrality_and_budgets}
	Let $\grphunweighted$ be a given underlying topology and suppose that Assumption~\ref{asmp:self_loops_in_underlying_ntw} holds. Consider the network formation game $\Mc{G}$ with Katz centralities given in~\eqref{eq:katz_score}. Let $\boldw^* \in \neset$, and let $\grph^*:=\grph(\boldw^*)$ denote the Nash equilibrium network induced by it. If $\grph^\prime :=(\agt^\prime,\edg^\prime,\adjmat^\prime)$ is an SCC of $\grph^*$, then $\exists \alpha \geq 0$ and $\exists \gamma>0$ such that $\katzscorei(\boldw^*)=\alpha,$ $B_i=\gamma$, $\forall i \in \agt^\prime$. Additionally, if $|\agt^\prime|\geq 2$ and $\grph^{\prime\prime}:=(\agt^{\prime\prime}, \edg^{\prime\prime}, \adjmat^{\prime\prime})$ is another SCC of $\grph^*$ such that $\wij^*>0$ for some $i \in \agt^\prime$ and $j \in \agt^{\prime\prime}$, then $\katzscorek(\boldw^*)=\alpha, B_k = \gamma,  \forall k \in \agt^{\prime\prime}$.
\end{theorem}
\begin{proof}
	Under the stated assumptions, notice that for any agent $i \in \grph^\prime$, there exists a directed cycle starting at $i$ and terminating at $i$ while passing through every node in $\agt^\prime$. First part of claim now follows from Lemma~\ref{lem:neighbors_have_higher_katz_at_nash_eqm} and Theorem~\ref{thm:nash_set_characterization}. Next, if there is an edge from $\grph^{\prime}$ to $\grph^{\prime\prime}$, then $\exists i \in \agt^\prime$ and $\exists j \in \agt^{\prime\prime}$ such that $\wij^* >0$. Since $|\agt^\prime|\geq 2$, there is an edge from $i$ to an agent $k \in \agt^\prime$ i.e., $\wik^* >0$. The result now follows from the first part, Lemma~\ref{lem:properties_of_better_resp} and Theorem~\ref{thm:nash_set_characterization}.
\end{proof}

Lemma~\ref{lem:neighbors_have_higher_katz_at_nash_eqm} states that in a Nash equilibrium network, any agent forms connections only with agents whose centralities are greater than or equal to its equilibrium centrality. Thus, a hierarchical structure emerges in the condensation graph of the Nash equilibrium network with respect to agents' Katz centralities. Further, Theorem~\ref{thm:agents_in_scc_have_same_centrality_and_budgets} states that all agents within a SCC of the Nash equilibrium network have the same budget and centrality. Moreover, if a SCC with at least two agents is connected to another SCC, then all agents in both SCCs have the same budget and centrality. 

\section{Simulations}

In this section, we present simulations to illustrate some of our analytical results. All simulations were run in \textsc{MATLAB}. We consider 10 agents with their Katz centralities as defined in~\eqref{eq:katz_score}. In the first set of simulations, we simulate BRD with random agent selection at every time step. The underlying topology $\grphunweighted$ is as shown in Figure~\ref{fig:self_loops} and contains self-loops at each node. Thus, Assumption~\ref{asmp:self_loops_in_underlying_ntw} holds. The budget vector containing $B_i$'s is $B \approx \begin{bmatrix}
	0.89\;\boldone_{3}^\top 
	& 0.17\;\boldone_{3}^\top 
	& 0.3 \;\boldone_{3}^\top
	& 0.86
\end{bmatrix}^\top$. The non-zero allocation values at $\boldw^*\in \neset$ attained under BRD are as follows: $w_{11}^*\approx 0.17$, $w_{12}^*\approx0.71$, $w_{22}^*\approx0.1$, $w_{23}^*\approx0.79$, $w_{31}^*\approx0.81$, $w_{33}^*\approx0.07$, $w_{42}^*\approx0.09$, $w_{43}^*\approx0.08$, $w_{59}^*\approx0.17$, 
$w_{62}^*\approx0.09$, $w_{63}^*\approx0.09$
$w_{73}^*\approx0.3$, $w_{82}^*\approx0.3$,
$w_{91}^*\approx0.3$, $w_{10,1}^*\approx0.62$, $w_{10,3}^*\approx0.24$. The corresponding Nash equilibrium network $\grph(\boldw^*)$ is shown in Figure~\ref{fig:self_loops}. The above data can be used to verify claims in Lemma~\ref{lem:properties_of_better_resp}. The evolution of centralities under BRD as shown in Figure~\ref{fig:evolution_of_centralities} is monotonic verifying claim~\ref{lem:btr_resp_1} in Lemma~\ref{lem:properties_of_better_resp}. The equilibrium centrality vector is $\katzscorestacked(\boldw^*) \approx \begin{bmatrix}
	7.82\;\boldone_{3}^\top 
	& 1.51
	& 0.62
	& 1.51
	& 2.64 \;\boldone_{3}^\top
	& 7.57
\end{bmatrix}^\top$. Here, $\boldone_3$ denotes a 3 dimensional vector of all ones. Using the above data, it can be verified that $\boldv(\katzscorestacked(\boldw^*))=\katzscorestacked(\boldw^*)$ thus $\katzscorestacked(\boldw^*)$ is a fixed point of $\boldv(\cdot)$ in~\eqref{eq:v_function} verifying that $\boldw^* \in \neset$ as suggested by Theorem~\ref{thm:nash_set_characterization}. The data also verifies the claims in Theorem~\ref{thm:agents_in_scc_have_same_centrality_and_budgets}.
\begin{figure}[h]
	\centering
	\begin{tabular}{ccc}
		\includegraphics[trim = 2.2in 1.0in 2.5in 0.1in, scale=0.14]{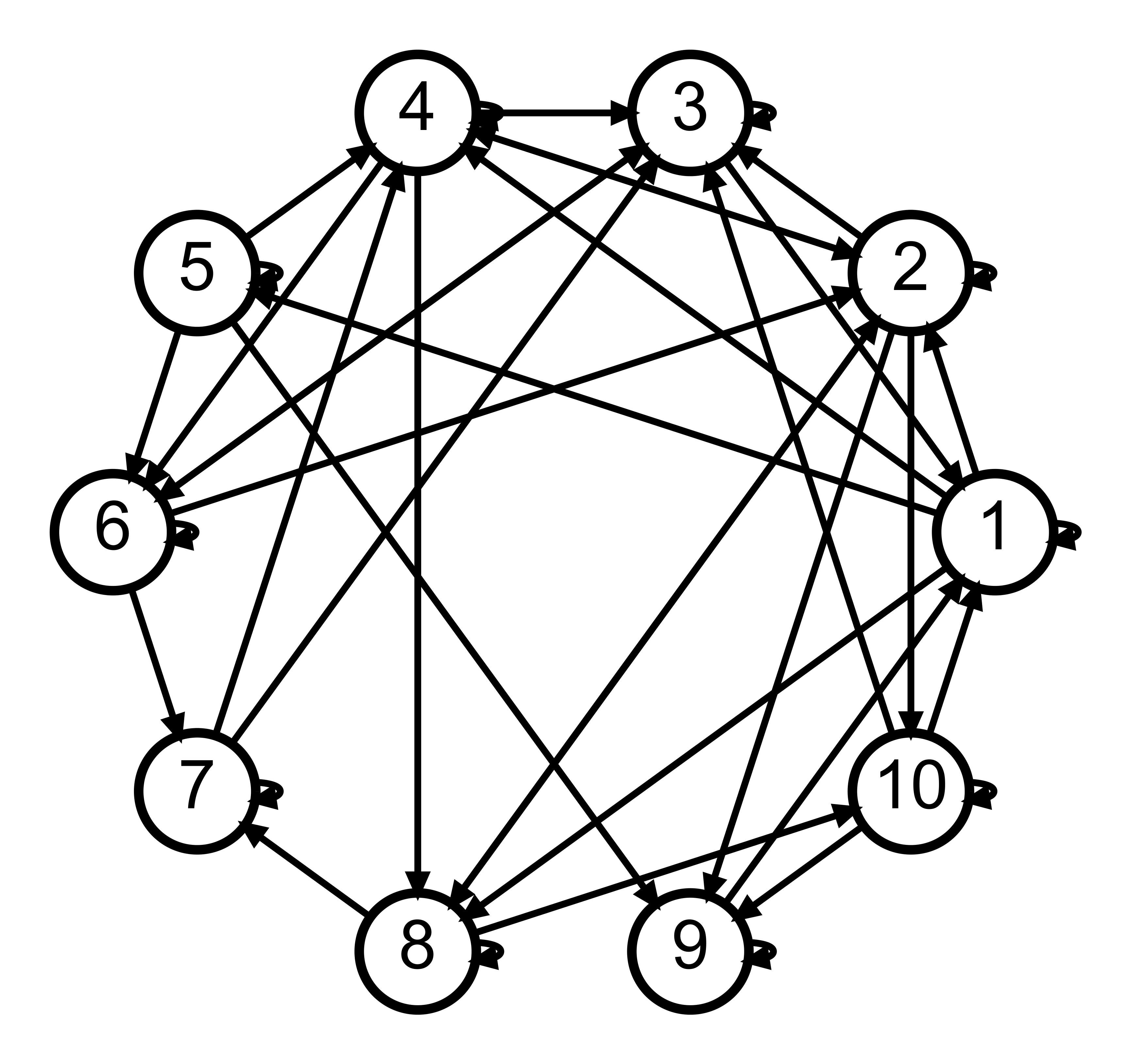} &
		\hspace{0.7in} 
		\includegraphics[trim = 4.2in 0.2in 3.5in 1.0in, scale=0.125]{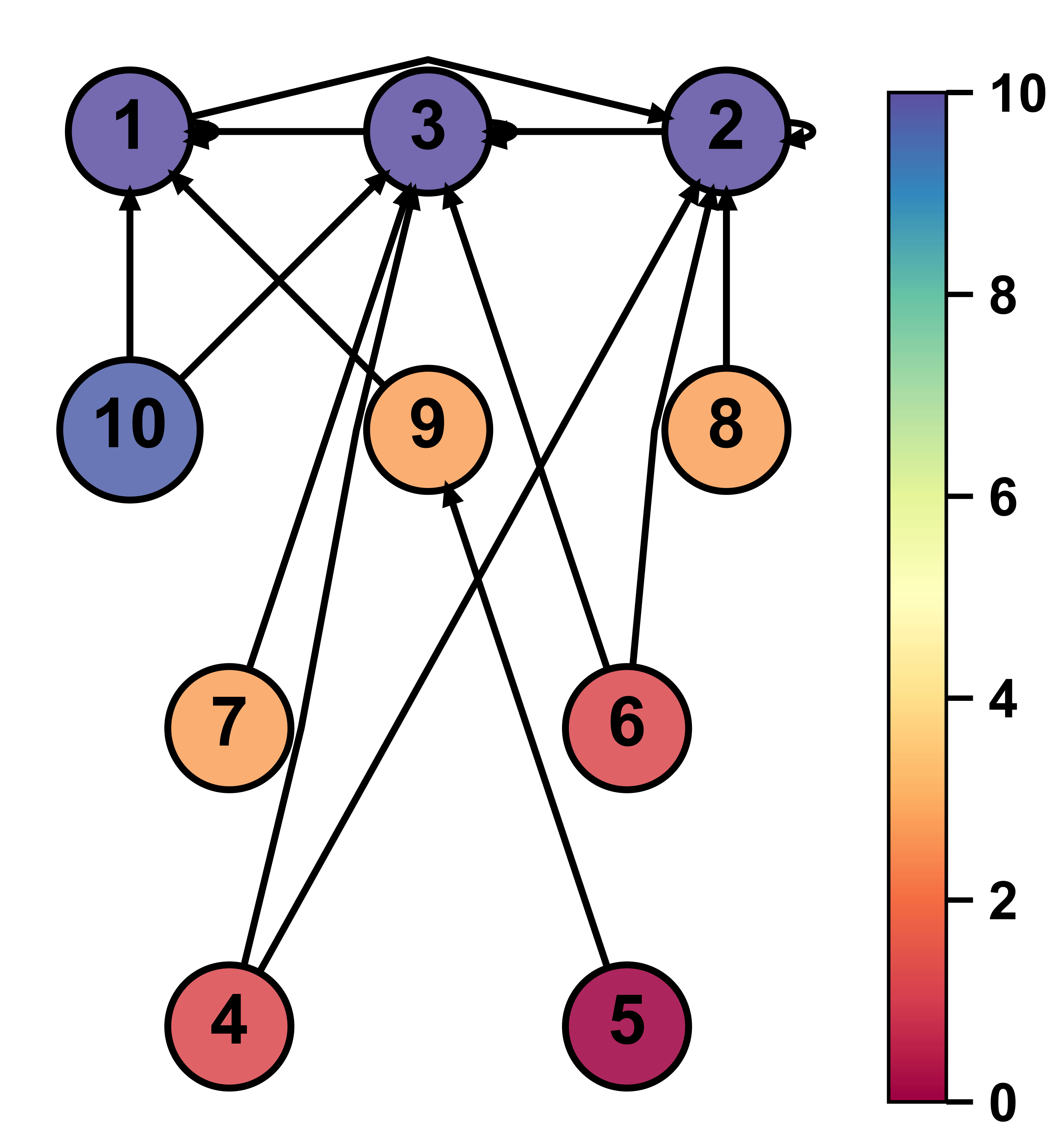} &
	\end{tabular}
	\caption{\hspace{-6pt}Convergence of BRD. (Left) Underlying topology with self-loops at every node. (Right) Unweighted Nash equilibrium network attained by BRD with self loops at nodes $\{1,2,3\}$. The color bar represents the continuum of centrality values in $[0,10]$.}
	\label{fig:self_loops}
\end{figure} 

\begin{figure}[t]
	\vspace{17pt}
	\centering
	\includegraphics[trim = 4in 0.4in 4.2in 0.4in, scale=0.3]{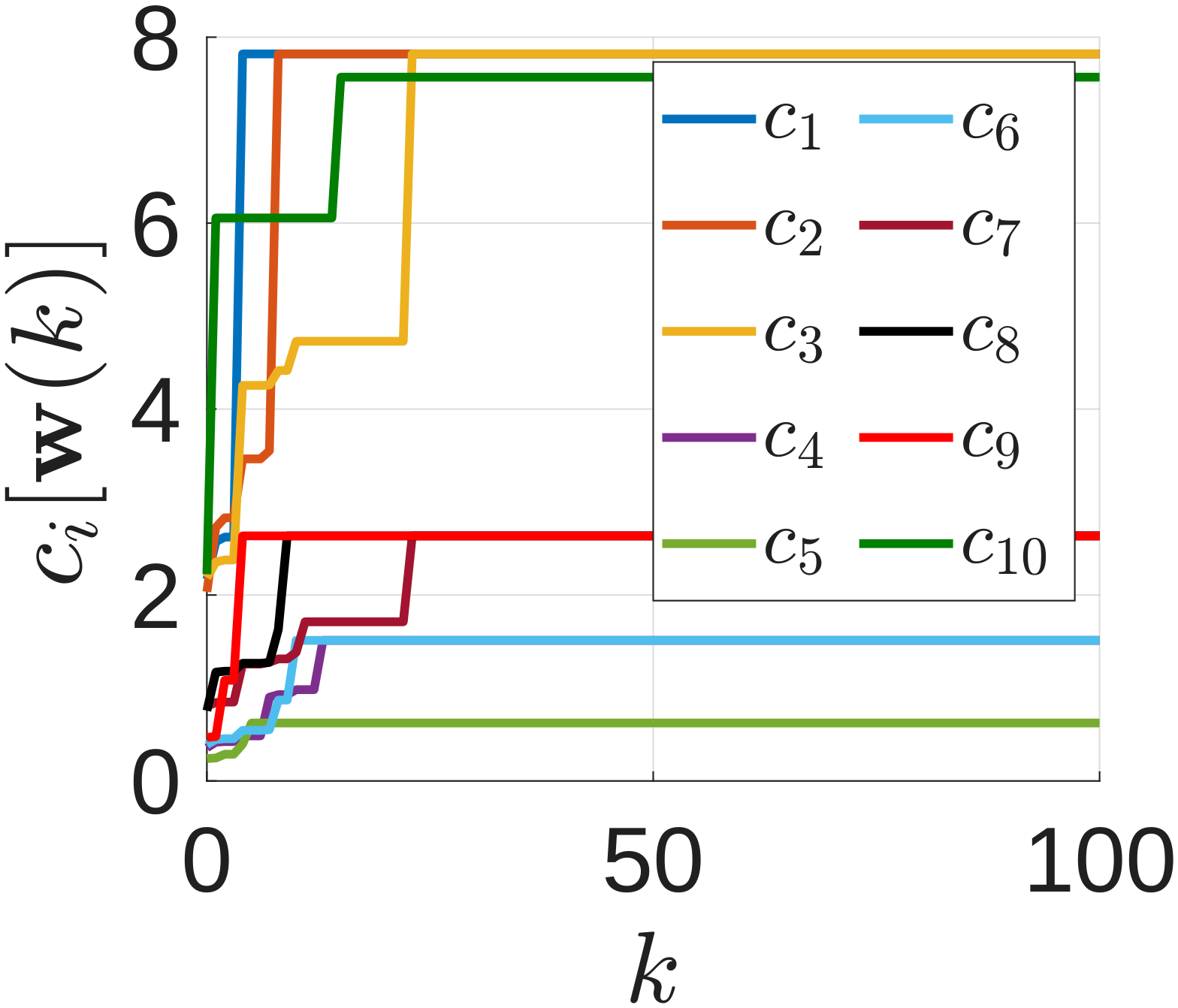}
	\caption{Evolution of $c_{i}$'s under BRD with $\grphunweighted$ as shown in Figure~\ref{fig:self_loops}.}
	\label{fig:evolution_of_centralities}
\end{figure}

In the second set of simulations, we assume the underlying topology $\grphunweighted$ to be complete. We again simulate BRD with random agent selection at each time step. Figure ~\ref{fig:complete_underlying_topology} shows the (unweighted) Nash equilibrium topology attained under BRD. We omit the allocation data $\boldw^*$ due to space constraints. The vector containing the budgets of all agents is $B \approx \begin{bmatrix}
	0.2\;\boldone_{3}^\top 
	& 0.83\;\boldone_{3}^\top 
	& 0.69 \;\boldone_{3}^\top
	& 0.17
\end{bmatrix}^\top$. The Nash equilibrium centrality vector is $\katzscorestacked(\boldw^*) \approx \begin{bmatrix}
	1.15\;\boldone_{3}^\top 
	& 4.77\;\boldone_{3}^\top 
	& 3.98\;\boldone_{3}^\top
	& 0.98
\end{bmatrix}^\top$. The above data and the plot on the right side of Figure~\ref{fig:complete_underlying_topology} verify the claims in Theorem~\ref{thm:nash_eqm_properties_complete_underlying_ntw} and the observation made below it.
\begin{figure}[h]
	\centering
	\begin{tabular}{ccc}
		\includegraphics[trim =2.1in 0.1in 5.5in 0.1in, scale=0.25]{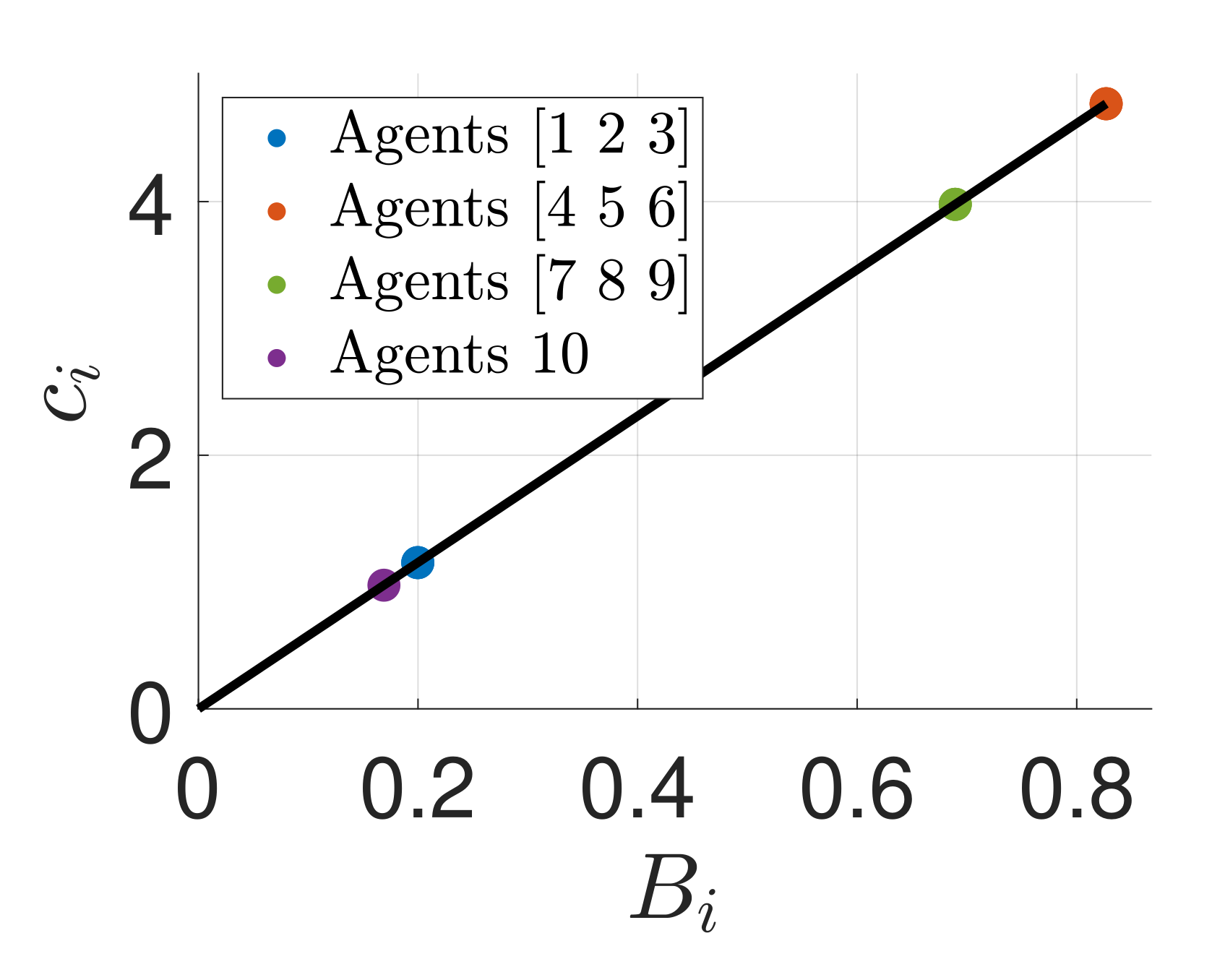} &
		\hspace{1.2in} 
		\includegraphics[trim = 0in 0.2in 3.9in 0.2in, scale=0.09]{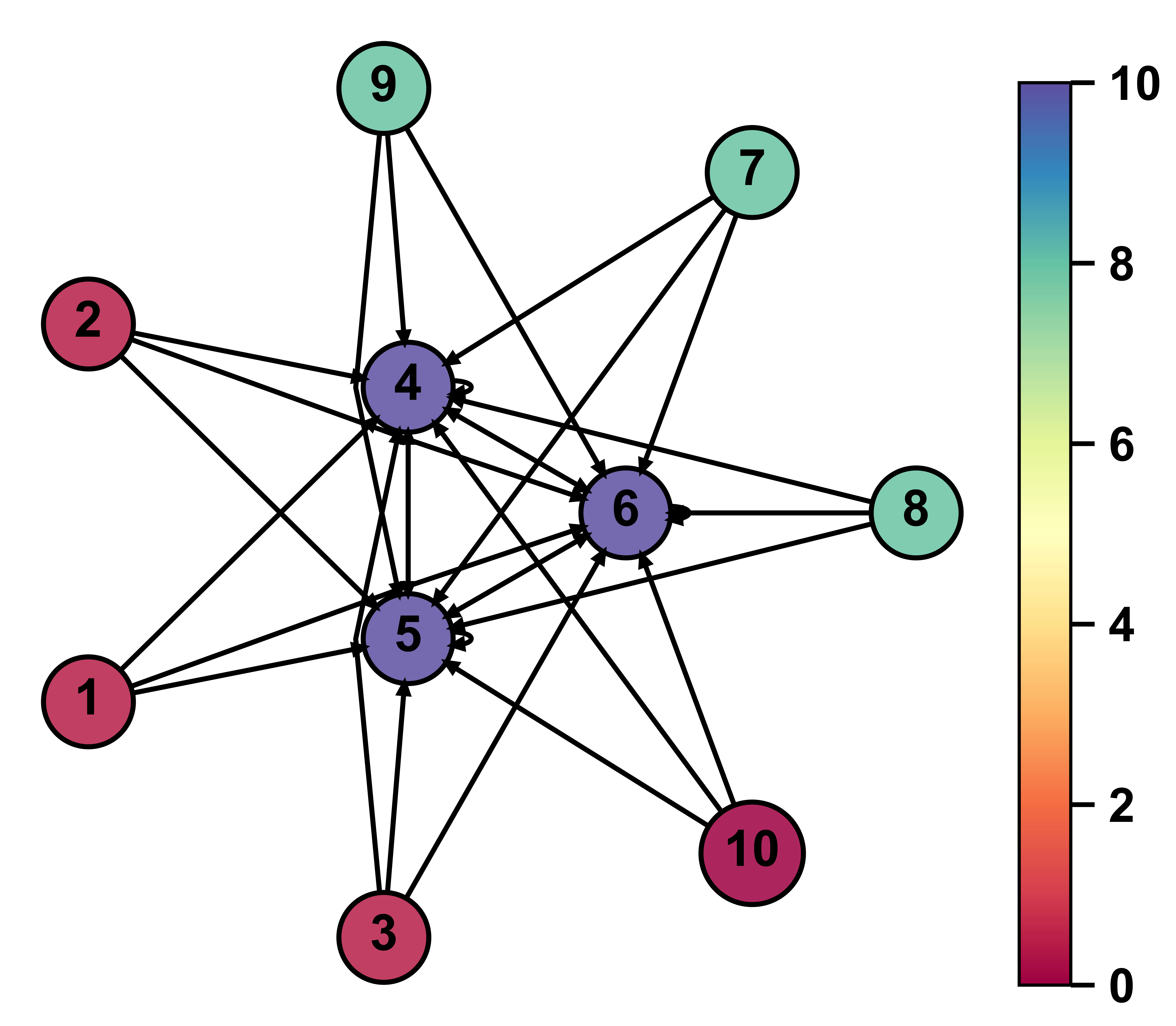} &
	\end{tabular}
	
	\caption{\hspace{-6pt} When $\grphunweighted$ is complete. (Left) Plot of centralities vs agents' budgets. (Right) Unweighted Nash equilibrium network reached under BRD. Self loops exist at nodes $\{4,5,6\}$. The color bar represents the continuum of centrality values in $[0,10]$.}
	\label{fig:complete_underlying_topology}
\end{figure} 

\section{Conclusion}

In this paper, we have studied a strategic network formation game where agents form directed weighted networks to maximize their Katz centrality. We have provided necessary and sufficient conditions for a network to be a Nash equilibrium and have characterized the set of Nash equilibrium networks under special underlying topologies. We have also shown that unilateral better responses at Nash equilibria are still Nash equilibria. We have shown that sequential BRD converges to the set of Nash equilibria. Finally, we have provided simulation results to verify our theoretical findings. Future work includes further analysis of the Nash equilibrium networks for different classes of underlying topologies and budget constraints, modeling and analysis of bounded rationality in the network formation game and network formation process in terms of limited information and computational limitations of the agents.

\bibliographystyle{IEEEtran}
\bibliography{references}

\end{document}